\documentclass[twocolumn]{article}

\usepackage{times}
\usepackage{soul}
\usepackage{url}
\usepackage[hidelinks]{hyperref}
\usepackage[utf8]{inputenc}
\usepackage[small]{caption}
\usepackage{graphicx}
\usepackage{amsmath}
\usepackage{amsthm}
\usepackage{booktabs}
\usepackage{algorithm}
\usepackage{algorithmic}
\usepackage{halloweenmath}
\usepackage[capitalize]{cleveref}
\usepackage{amssymb, mathtools}
\usepackage{etoolbox}
\usepackage[dvipsnames]{xcolor}
\usepackage{comment}
\usepackage{enumerate}
\usepackage{enumitem}
\usepackage[normalem]{ulem}
\usepackage{xspace}
\usepackage{bbm}
\usepackage{fdsymbol}
\usepackage{yfonts}
\usepackage{upgreek}
\usepackage{utfsym}
\usepackage{lipsum}
\usepackage{graphicx}% http://ctan.org/pkg/graphicx
\usepackage{tikz}
\usepackage{tikzsymbols}

\usepackage{footmisc}

\usepackage{geometry}
\theoremstyle{plain}
\newtheorem{theorem}{Theorem}
\newtheorem{corollary}[theorem]{Corollary}
\newtheorem{lemma}[theorem]{Lemma}

\newtheorem{fact}[theorem]{Fact}

\newtheorem{observation}[theorem]{Observation}

\newcommand{\defqedsymbol}{\ensuremath{\lozenge}}
\theoremstyle{definition}
\newtheorem{definition}[theorem]{Definition}
\AtBeginEnvironment{definition}{\pushQED{\null\hfill\defqedsymbol}}
\AtEndEnvironment{definition}{\popQED}%\null\hfill\ensuremath{\lozenge}}

\theoremstyle{remark}

\definecolor{specificationblue}{RGB}{60,70,100}
\newcommand{\specificationqedsymbol}{\ensuremath{\lozenge}}
\newtheoremstyle{specificationstyle}
  {\topsep}   % space above
  {\topsep}   % space below
  {\normalfont} % body font
  {}          % indent
  {\bfseries} % head font
  {\textcolor{black}{.}}         % punctuation after theorem head
  { }         % space after theorem head
  {\textcolor{black}{\thmname{#1}\thmnumber{ #2}}%
   \thmnote{ {\normalfont(#3)}}}

\theoremstyle{specificationstyle}
\newtheorem{specification}[theorem]{Specification}
\AtBeginEnvironment{specification}{\pushQED{\null\hfill\specificationqedsymbol}}
\AtEndEnvironment{specification}{\popQED}%\null\hfill\ensuremath{\lozenge}}
\newcommand{\db}{{\mathcal D}}

\newcommand{\dddb}{\db}

\newcommand{\dbmf}{{\mathfrak D}}
\newcommand{\dbmbb}{{\mathfrak D}_0}

 \newcommand{\ontologyy}{\mathcal{O}}

\newcommand{\DLCL }{DL-Lite$_{\text{\tiny \sc core}}$\xspace}
\newcommand{\DLCLP}{DL-Lite$_{\text{\tiny \sc core}}^+$\xspace}

\newcommand{\cq}{\text{CQ}}

\newcommand{\cqneg}{\text{CQ}^{s\neg}}

\newcommand{\ontol}{{\mathcal O}}

\newcommand{\thue}{{\mathcal G}}

\newcommand{\szymeq}{\simeq_\Pi}%{\xspace{\simeq_\Pi}\xspace}

\newcommand{\VVV}{\mathfrak V}

\newcommand{\strz}[3]{#1 \stackrel{#2}{\longrightarrow} #3 }

\newcommand{\jest}[1]{#1^{\text{\scalebox{0.8}{$\exists$}}}}

\newcommand{\abox}{\text{ABox}}
\newcommand{\tbox}{\text{TBox}}

\newcommand{\freeze}[1]{\text{\small \sc freeze}(#1)}
\newcommand{\freezed}{\text{\small \sc freeze}}
\newcommand{\normalfreezed}{\text{\sc freeze}}

\newcommand{\dom}[1]{\operatorname{Dom}(#1)}
\newcommand{\type}[2]{type(#1,#2)}

\newcommand{\pair}[1]{\langle #1 \rangle}
\newcommand{\set}[1]{\{ #1 \}}

\newcommand{\Thue}{{\mathbb G}}

\newcommand{\III}{\mathbb I}

\newcommand{\MMM}{\mathbb S}
\newcommand{\iii}{\Bbbk}

\newcommand{\love}{{\Wintertree}}
\newcommand{\dblove}{\db_\love}

\newcommand{\orange}[1]{\textcolor{orange}{#1}}

\newcommand{\iffi}{\textit{iff} }

\definecolor{goldl}{RGB}{171,126,1}

\newcommand{\pphi}{{\mathbb Q}}
\newcommand{\newaboxneg}{{\mathbb A}}

\newcommand{\ontolneg}{{\mathfrak K}} 

\newcommand{\aabox}{{\mathfrak A}} 
\newcommand{\ttbox}{{\mathfrak T}} 

\newcommand{\psineg}{{\mathfrak Q}}

\newcommand{\unaryA}{{\mathbf A}}

\newcommand{\size}{\mathsf{size}}
\newcommand{\problem}{\Thue}
\newcommand{\start}{\mathsf{L}}
\newcommand{\finish}{\mathsf{R}}
\newcommand{\blank}{\mathsf{B}}
\newcommand{\es}{\mathsf{Z}}
\newcommand{\esp}{\mathsf{Z}'}
\newcommand{\szym}{\sim_\Pi}
\newcommand{\bear}{\mathbb{Q}}
\newcommand{\beari}{\mathbb{Q}_i}

\newcommand{\LL}{\mathsf{L}}
\newcommand{\RR}{\mathsf{R}}
\newcommand{\diam}{{\textcolor{gray}{\vardiamondsuit}}}
\newcommand{\idiam}{{\vphantom{a}\!\diam}}

\newcommand{\cpred}{\mathsf{C}}
\newcommand{\dpred}{\mathsf{D}}
\newcommand{\epred}{\mathsf{E}}
\newcommand{\fpred}{\mathsf{F}}

\newcommand{\xpred}{\mathsf{X}}
\newcommand{\ypred}{\mathsf{Y}}
\newcommand{\zpred}{\mathsf{Z}}
\newcommand{\ypredp}{\mathsf{Y}'}
\newcommand{\ppred}{\mathsf{P}}
\newcommand{\spred}{\mathsf{S}}

\newcommand{\pillow}{Pillow}

\newcommand{\szymeqrep}[1]{[#1]_{\szymeq}}

\newcommand{\korz}{Root}
\newcommand{\tooclose}{2Close}
\newcommand{\vlove}{\love(\db)}

\title{Someone slept in my bed!
On the entailment problem for
conjunctive queries with safe negation
over
DL-Lite$_{core}$ knowledge bases}
\author{
  Jerzy Marcinkowski\\
  University of Wrocław\\
  \texttt{jma@cs.uni.wroc.pl}
  \and
  Piotr Ostropolski-Nalewaja\\
  University of Wrocław\\
  \texttt{postropolski@cs.uni.wroc.pl}
}
\begin{document}

\maketitle

\begin{abstract}
We solve a long standing open problem, showing  that the query answering for 
conjunctive queries with safe negation, over
\DLCL knowledge bases, is undecidable.
%We actually prove this negative result even for a certain logic weaker than \DLCL.
\end{abstract}

\section{Preliminaries. Part 1.}\label{sec:preliminaries}

%\orange{\sout{In this Section we very briefly present some standard notions, to make sure that the authors and the Readers are on the same page.}}

We assume that the Reader is familiar with the standard notions regarding first-order logic (FOL) and its semantics.
In particular, we use the terms {\em signature}, and {\em atomic formula} (or {\em atom}) in the standard sense. 
{\em Fact} is an atomic formula without variables (that is, only {\em constants} are allowed).
We assume that there exists some fixed infinite set $\mathcal C$ of constants.

\subsection{Structures and Relations.}
The signatures
we consider only comprise unary and binary {predicates}, so we think of
{\em relational structures} (for short just {\em structures})
over such signatures as (coloured and labelled) graphs and, in consequence, we call
elements of their domains {\em vertices} and sometimes we call the binary atoms {\em edges}.

We find it notationally convenient to think that { structures} are just sets of facts.
This means that the structures we consider never have isolated vertices and that
only elements of $\mathcal C$ can serve as verticies. This can be assumed,
for the purpose of this paper, without loss of generality.

For a structure $\db$ by $\dom{\db}$ we denote the set of its vertices (that is the set of
constants the facts of $\db$ mention). Structure $\db_1$ is called a {\em substructure} of $\db_2$ if
$\db_1\subseteq\db_2$ (recall they are sets of facts).
If $\db$ is a structure and $C\subseteq \dom{\db}$ then the set:
$$\db[C] = \{E(c,d)\in \db: \;  c,d\in C\} \cup \{A(c)\in \db:  c\in C\}  $$
%
%$\orange{\set{A(\vec{a}) \in \db \mid \vec{a} \in C \lor \vec{a} \in C^2}}$
is called an {\em induced} by $C$ {\em substructure} of $\db$.

%, denoted as $\db[C]$.

%$$\{E(c,d): \; E(c,d)\in \db \;\wedge\; c,d\in C\} \cup \{A(c): \;A(c)\in \db \;\wedge\; c\in C\}  $$

\subsection{Queries.}\label{sec:sugar1}
A \emph{conjunctive query} (CQ) is a FOL formula which is a conjunction of atoms preceded by existential quantifiers.
A \emph{conjunctive query with safe negation} ($\cqneg$) is a  CQ where negated atoms are allowed. {It is however assumed, that every variable that appears in such query, must appear in
some positive (that is non-negated) atom}. A CQ (or a $\cqneg$) is {\em Boolean} if all its variables are bound by existential quantifiers.
We use the notations $\cq$ (or $\cqneg$) also  to denote the respective classes of queries.

For a Boolean $\xi\in \cq \cup\cqneg$ and a structure $\db$ by $\db\models \xi$
we mean that $\xi$ is true in 
$\db$, in the usual FOL sense.
Recall that if $\xi\in \cq$ then $\db\models \xi$ is equivalent to
the existence of a mapping from (the variables of) $\xi$ to $\db$ which is a homomorphism. If $\xi\in \cqneg$ then,
for $\db\models \xi$ to hold, this homomorphism must satisfy some additional natural constraint, namely
the image of a negated atom in $\xi$ must not be an atom in $\db$.
We will call such mappings {\em satisfying mappings}.

\subsection{The logic \texorpdfstring{\DLCL}{DL-lite-core} and query entailment.}
The DL-Lite family of logics, and the logic \DLCL in particular,  were introduced in
\cite{2007-introduces-DL-lite-core_tractable-CQ-answering-in-DLlite}.
Since our work is technical in nature, we will not discuss DL-Lite family in the full generality here, focusing instead on presenting the minimal set of definitions allowing us to formulate our results.

%Their definitions use the DL jargon, relying on the notions of a {\em role} and a {\em concept}. 
A signature $\Sigma$ comprising unary and binary {predicates} is assumed. A {\em role} is  any binary relation {over} $\Sigma$ or its inverse. A {\em concept} is any unary relation {over} $\Sigma$ or any unary relation that can be defined as a projection of a role: for a role $R$, the set $\{x : \exists y\, R(x,y) \}$ is a concept, denoted as $\exists R$.

A DL-Lite knowledge base $\ontol$ is a pair $\pair{\tbox, \abox }$,
where $\abox$ is a structure (i.e. a set of facts) and
 $\tbox$ is the ``domain knowledge'' -- a set of
expressions involving  roles and concepts. \DLCL is the least expressive of the logics of  family
and  only  two forms of such expressions ale allowed in a \DLCL \tbox, namely:
$$
(\usym{2660})\;\; B \sqsubseteq B' \;\;\;\;\;\;\;\;\;\;\;\;\;\;\;\;
(\usym{2665})\;\; B \sqsubseteq \neg B'
$$
\noindent
 where $B$ and $B'$ are concepts. The semantics of the expressions is natural, with
 $\sqsubseteq$ meaning subset, and $\neg$ meaning complementation. Notice that
 (\usym{2665}) means that two concepts are disjoint.

 For a structure $\db$ and a knowledge base $\ontol=\pair{\tbox, \abox}$ we say that $\db$ is a model of $\ontol$, denoted as $\db \models \ontol $ when
 $\abox\subseteq \db$ and all the inclusions from $\tbox$ are satisfied in $\db$ in the usual FOL sense.

 For a Boolean conjunctive query $\xi$ (possibly with safe negation), and a knowledge base $\ontol$, by $\ontol \models \xi$ (``$\ontol$ {\em entails} $\xi$'') we mean that $\xi$ is true in
every structure $\db$ such that $\db\models \ontol$.

%%%%%%%%%%%%%%%%%%%%%%%%%%%%%%%%%%%%%%%%%%%%%%%%%%%%%%%%%%%%%%%%
\section{Our contribution}
%%%%%%%%%%%%%%%%%%%%%%%%%%%%%%%%%%%%%%%%%%%%%%%%%%%%%%%%%%%%%%%%

Our main technical result is:

\begin{theorem}\label{main}
The following problem is undecidable:

\noindent
{Given a \DLCL knowledge base $\ontologyy$ and a Boolean conjunctive query with safe negation $\phi$. Does  $\ontologyy\models \phi$?
}
\end{theorem}

%%%%%%%%%%%%%%%%%%%%%%%%%%%%%%%%%%%%%%%%%%%%%%%%%%%%%%%%%%%
%%%%%%%%%%%%%%%%%%%%%%%%%%%%%%%%%%%%%%%%%%%%%%%%%%%%%%%%%%

\section{Preliminaries. Part 2.}

In this section, we present some non-standard notations, which will first let us equivalently restate Theorem \ref{main} and then will be very useful in Sections \ref{sec:implementation} and \ref{sec:dowody} where our main technical result is proven.

\subsection{CQs with one free variable.}

For a conjunctive query $\xi(x)$, possibly with  safe negation, with one free variable, we denote by $\jest{\xi}$ the Boolean conjunctive query $\exists x\,\xi(x)$.
 Variable $x$ will
 be called {\em the distinguished variable} in $\jest{\xi}$.

To see what this notation can be good for imagine two queries  $\xi_1(x_1)$ and  $\xi_2(x_2)$, each with one free variable and a structure $\db$.
Consider a new boolean conjunctive query:
$$ \xi \;\;\; = \; \;\; \exists x_1,x_2\;\;\; \neg H(x_1,x_2) \;\; \wedge \;\; \xi_1(x_1) \;\;  \wedge  \;\; \xi_2(x_2) $$
Now, to satisfy $\xi$ in $\db$, one needs to find
%Then (and this may seem to be an obvious observation, but will be very useful)
%in Section \ref{secmain}
%in order to satisfy query $\xi$, in some structure $\dddb$,
two satisfying mappings $h_1: \jest{\xi_1} \rightarrow \dddb $ and $h_2: \jest{\xi_2} \rightarrow \dddb$ % need to be found,
such that $H(h_1(x_1),h_2(x_2))\not\in \db$. Apart from the last constraint, which concerns the
 distinguished variables of the two queries,
the mappings are independent,
because even if (syntactically)  $\jest{\xi_1}$ and $\jest{\xi_2}$  shared some variables, the shared variables
are local in $\xi_1(x_1)$ and in  $\xi_2(x_2)$.
%
%In such context, we will call queries $\jest{\xi_1}$ and $\jest{\xi_2}$ {\em components} of $ \xi$.

\subsection{Types of vertices and the \texorpdfstring{$\normalfreezed$}{freeze} construct.}

For a structure $\db$ over the signature $\Sigma$ and for a vertex $c\in \dom{\db}$ we define
the (unary) \emph{type} $\type{c}{\db}$ of $c$ in $\db$ in the natural way:

\noindent
\begin{alignat*}{3}
 \pair{ \;\;\{A&\in\Sigma \mid &&\db\models A(c)\},\\
        \{R&\in\Sigma \mid \exists x\; &&\db\models R(c,x) \},\\
        \{R&\in\Sigma \mid \exists x\; &&\db\models R(x,c)\}    \;\;}    
\end{alignat*}

\noindent
In words, $\type{c}{\db}$ tells us which unary predicates are true in $c$,  which binary predicates lead to $c$, and which leave $c$.

All \DLCL is able to say is about types.  It is very easy to notice that:

\begin{observation}\label{obs:trivial}
Let $\tbox$ be a set of statements of  the form (\usym{2660}) and (\usym{2665}).
Suppose ${\db}$ and ${\db}'$ are two structures, such that  $\dom{\db}=\dom{\db}'$ and that for each
$c\in \dom{\db}$ it holds that $ \type{c}{\db} = \type{c}{\db'}$. Then $\db\models\tbox$ if and only if  $\db'\models\tbox$.
\end{observation}

So, as an example, suppose:
$\db=\{E(a,b),E(b,c)\}$ and $\db'=\{E(a,b),E(b,c),E(a,c)\}$.
Then no \DLCL $\tbox$ can distinguish between the two.
This is a bit problematic from the point of view of proof of Theorem \ref{main}, because at some point in the proof we need to make sure that certain conjunctive query {\bf with safe negation} is
satisfied in every model. To this end we need to find a way to make sure that, in each model of the knowledge base we are going to construct,
there is always enough of {\em negativity}  to satisfy the negated atoms, that is there are enough atoms that are not true in the model.

But, as Observation \ref{obs:trivial} says we cannot prohibit the presence of atoms that do not change
the types of the vertices. So let us at least  prohibit the presence of atoms that {\bf do  change
the types of the vertices}. And this is what the $\freezed$ construct is about, which leads to the definition of a new logic,
that we call \DLCLP:

\begin{definition}[\DLCLP syntax] A \DLCLP knowledge base is a triple $\pair{\tbox, \abox, \freeze{\mathbb A}}$ such that:\\
\phantom{a} \textbullet\quad $\pair{\tbox, \abox}$ is a \DLCL knowledge base, and\\
\phantom{a} \textbullet\quad ${\mathbb A}\subseteq \abox$.
\end{definition}

\begin{definition}[\DLCLP semantics]
For a \DLCLP knowledge base $\ontol=\pair{\tbox, \abox, \freeze{\mathbb A}}$ and a structure $\db$ we say that $\db\models \ontol$ \iffi\!:\\
\phantom{a}\textbullet \quad $\db\models \pair{\tbox, \abox}$, and \\
\phantom{a}\textbullet \quad $\type{c}{\db}=\type{c}{\mathbb A}$ holds for all $c\in \dom{\mathbb A}$
\end{definition}

Now, proving the following lemma is quite straighforward:

\begin{lemma}\label{lem:freeze-redukcja}
For a \DLCLP knowledge base {$\ontol=\pair{\tbox, \abox, \freeze{\mathbb A}}$} and query $\psi\in \cqneg$ one can effectively construct a \DLCL knowledge base $\ontol'$ such that $$\ontol\models\psi \Leftrightarrow \ontol'\models\psi.$$
\end{lemma}
\begin{proof}\!\!(sketch)\,
    {We define $\ontol' =\pair{\abox', \tbox'}$ as follows. Let $\abox' = \abox \cup \set{\mathrm{is}_c(c) \mid c\in \dom{\mathbb{A}}}$, where $\mathrm{is}_c$ is a fresh unary predicate for a constant $c$; note that $\dom{\mathbb{A}}$ is finite. Define $\tbox'$ as $\tbox$ with expressions $\mathrm{is}_c \sqsubseteq \neg B$ for each concept $B$ that is not satisfied in $c$ in $\abox$.}
\end{proof}

In view of Lemma \ref{lem:freeze-redukcja} we equivalently restate Theorem \ref{main} as:

\begin{theorem}\label{mainp}
The following problem is undecidable:

\noindent
{Given a \DLCLP knowledge base $\ontologyy$ \; and a Boolean conjunctive query with safe negation $\phi$. Does  $\ontologyy\models \phi$?
}
\end{theorem}

And it is Theorem \ref{mainp} that we prove in Sections \ref{sec:high-level} and \ref{sec:implementation}.

%%%%%%%%%%%%%%%%%%%%%%%%%%%%%%%%%

\section{High level proof of Theorem \ref{mainp}}\label{sec:high-level}

In this section we present the high level part of our proof of Theorem \ref{mainp}. 
Some (or maybe even most)
of the objects needed in the proof will remain undefined in this section.
The list of such undefined notions includes:

\begin{enumerate}[label=(\roman*), itemsep=-0em, topsep=0.0em]
    \item a decision problem $\Thue$;
    \item two finite sets $\MMM$ of symbols and  $\III$ of names;
    \item a conjunctive query with safe negation $\pphi_i(x_i)$ with one free variable, for each ${i\in\III}$;
    %conjunctive query  with safe negation,
    
    \item a finite structure $\newaboxneg_i$, for each $i \in \III$.
\end{enumerate}

Instead of defining the above notions, in the current section we will just {\bf specify} their properties.\footnote{We use the {\em mathbb} font  for the objects that await definition.}
Then we will prove that, if they all were indeed defined in a
way satisfying the specifications,
then the theorem would follow.

Let us insist that what we do here is not an ``informal idea'',
which will be followed by a precise proof. What we are doing here
is totally rigorous, and --- we hope --- quite elegant. % while not too hard to read.
Then, in Sections \ref{sec:implementation} and \ref{sec:dowody} we show how to define the notions from the above list
in a way satisfying the specifications. {Definitions of the notions (ii)-(iv) will depend
on the particular instance $\thue$ of the problem $\Thue$ so, for example, instead of 
$\MMM$ we should write $\MMM_\thue$. For the sake of readability, we omit this subscript. This should not cause any confusion, since we will always consider one fixed $\thue$.} 
Unlike the proof in the current section, the construction, and argument,   in 
Sections \ref{sec:implementation} and \ref{sec:dowody}
 has a bit of a low-level programming flavour.

%%%%%%%%%%%%%%%%%%
\subsection{The source of undecidability}\label{subsec:source}

\begin{specification}\label{spec:of-problem}
Problem $\Thue$ will be undecidable and r.e.
%Each instance $\thue$ of 
%the problem
%$\Thue$ will be a \orange{pair consisting of a finite set $\MMM$ and a finite subset of $\MMM^* \times \MMM^*$.}
%finite  subset of $\MMM^* \times\; \MMM^*$ for some  finite set $\MMM$, which is also a part of the instance.
\end{specification}

\noindent
In order to prove Theorem \ref{mainp}  we are going to construct,
for
 each instance $\thue$ of the problem $\Thue$,
a \DLCLP
 knowledge base\footnote{We use the {\em mathfrak} font for the objects we construct {in the reduction}: the knowledge base, its elements, and the query.}
 $\ontolneg$ and a Boolean conjunctive query with safe negations $\psineg$ such that  the following 
 equivalence will holds:
\vspace{0.1mm}
\begin{center}
\hfill$\thue$ is a positive instance ~~~ \iffi  ~~~ $\ontolneg \models \psineg$ \hfill  $\bigpumpkin$
\end{center}

\noindent
From now on, {we consider an instance $\thue$ 
%--- and thus sets $\MMM$ and $\III$ --- 
to be fixed.}
%, which means that $\MMM$ is fixed too. 

%%%%%%%%%%%%%%%%%%%%
\subsection{Towards the knowledge base\texorpdfstring{ $\ontolneg$.}{.} Our \texorpdfstring{$\tbox$}{Tbox}. }

{We begin with the signature of the knowledge base $\ontolneg$:}
\begin{definition}
    The signature $\Sigma$ of  $\ontolneg$ comprises:
\begin{itemize}[itemsep=-0.2em, topsep=0.0em]
    \item three unary relations, $Root$, $\unaryA$, and $Pillow$,
    \item the set $\MMM$ (mentioned above) of binary relation symbols,
    \item a binary relation $2Close_i$ (``too close'')
    for each $i\in \III$.\qedhere %(where $\III$ is as in Specification \ref{spec:2}).
\end{itemize} \end{definition}

Now, recall that the $\ontolneg$ we are going to construct
is a  \DLCLP knowledge base, so it will be 
 a triple $\pair{\ttbox, \aabox, \freeze{\aabox_0}}$ for some $\tbox$ $\ttbox$, some  $\abox$ 
 $\aabox$, and some 
 $\aabox_0 \subseteq \aabox$. 

We are now ready to define our $\tbox$ $\ttbox$ :

\begin{definition}\label{def:tbox}
 $\ttbox$ contains the following statements:
\begin{itemize}
    \item [(i)]\!$\exists \xpred^{-1} \sqsubseteq \exists \ypred$  for each pair  $\xpred,\ypred\in\MMM$.
    \item [(ii)] \!$ Root \sqsubseteq \exists \xpred$ for all $ \xpred \in \MMM$.
    \item [(iii)]\!$ \exists \xpred^{-1} \sqsubseteq \neg \exists 2Close_i$ for all $\xpred\in \MMM$ and all $i\in \III$,
    \item [(iii)]\!$ \exists \xpred^{-1} \sqsubseteq \neg \exists 2Close_i^{-1}$ for all $\xpred\in \MMM$ and all $i\in \III$,
    \item [(iii)]\!$ \unaryA \sqsubseteq \neg \exists 2Close_i$ for all  $i\in \III$,
    \item [(iii)]\!$ \unaryA \sqsubseteq \neg \exists 2Close_i^{-1}$ for all  $i\in \III$.\qedhere
\end{itemize}
\end{definition}

\noindent
In human language this reads as:

\noindent
\begin{itemize}[itemsep=-0.2em, topsep=0.0em]
    \item [(h1)] A vertex reached by some edge $\xpred\in \MMM$ must produce an $\ypred$ for every $\ypred\in \MMM$;
    \item [(h2)] A vertex which is a $Root$ must produce an $\ypred$ for every $\ypred\in \MMM$;
    \item [(h3)] A vertex reached by any $\xpred \in \MMM$ (or one which is a $\unaryA$)  is never too close to any other vertex.
\end{itemize}

%%%%%%%%%%%%%%%%%%%%%%%%%%%%%%%%%%%%%%%%%%%%%%%%%%%%%%%%%%%%%%%%%%%
\subsection{Towards the \texorpdfstring{$\abox$}{ABox}. }

Now let us concentrate on the structures $\newaboxneg_i$, for  $i \in \III$.
We like to think of them as beds in which our little bears will sleep, with their heads duly resting on the $Pillows$. 
Their specification 
will come in two parts:

\begin{specification}[beds pt.~1]\label{spec:of-disjointness}
Structures $\newaboxneg_i, \newaboxneg_j$, for each $i,j \in \III$  will satisfy:
\begin{enumerate}
    \item if $i\neq j$ then  $\dom{\newaboxneg_i}\cap \dom{\newaboxneg_{j}}=\emptyset$;
    \item $\dom{\newaboxneg_i}\cap \dom{\newaboxneg}=\emptyset$;
      \item $a\not\in \dom{\newaboxneg_i}$\qedhere
\end{enumerate}
\end{specification}

\begin{definition}
Let:
\begin{align*}
    \aabox_0 = \textstyle\bigcup_{i\in \III} \newaboxneg_i \quad\;  \newaboxneg = \{Root(a), \unaryA(a)\} \quad\;  \aabox = \aabox_0 \cup \newaboxneg \;\tag*{\qedhere}
\end{align*}
%$$ \aabox_0 = \textstyle\bigcup_{i\in \III} \newaboxneg_i \quad  \newaboxneg = \{Root(a), \unaryA(a)\} \quad  \aabox = \aabox_0 \cup \newaboxneg \qedhere$$
%\begin{itemize}
%    \item [$\aabox_0$] = $\bigcup_{i\in \III} \newaboxneg_i$
%    \item [$\newaboxneg$] = $\{Root(a), \unaryA(a)\}$
%    \item [$\aabox$] = $\aabox_0 \cup \newaboxneg$\qedhere
%\end{itemize}
%$\aabox_0 =  \bigcup_{i\in \III} \newaboxneg_i  $;\\
%$\newaboxneg = \{Root(a), Ground(a)\} $;\\
%$\aabox = \aabox_0 \cup \newaboxneg $
\end{definition}

Notice, that our knowledge base is at this point defined: $$\ontolneg\;\;=\;\;\pair{ \ttbox, \aabox,  \;\freeze{\aabox_0}} $$

Now, suppose $\db\models \newaboxneg$ for some structure $\db$. We define $\love(\db)$ as the set of all the vertices of $\db$ reachable
from $a$ by a path of $\MMM$ symbols:

\begin{definition}
The set $\love(\db)$ is the inclusion-minimal set satisfying:
\begin{itemize}[topsep=0em,itemsep=0em]
    \item $a\in \love(\db)$;
    \item {if $u \in \love(\db)$ and $\db\models \xpred(u,v)$, then  $v\in\love(\db)$, for all} vertices $u,v$ {of $\db$} and $\xpred\in\MMM$.%, and if $u\in\love(\db)$, then  $v\in\love(\db)$.
\end{itemize}
\noindent
Then we denote as $\dblove$ the structure $\db[\love(\db)]$, that is the substructure of $\db$
induced by $\love(\db)$.
\end{definition}

One can easily get tempted to imagine $\dblove$
as an infinite tree, of out-degree equal to $|\MMM|$ and with root in $a$.
But reality is more complicated, and many strange things\footnote{As you soon will see, $\dblove$ is where undecidability comes from.} may happen here. In the extreme case, $\dblove$ could even
be just the vertex $a$ with all edges $\xpred$, which origin from $a$, looping back to it.

What we know for sure is that vertices of  $\love(\db)$ are never too close to anything:

\begin{lemma}\label{lem:no-hatred}
Suppose $\db\models \ontolneg$. If $\db\models 2Close_i(s,t)$  for some $i\in\III$, then
$s\not\in \love(\db)$ and $t\not\in \love(\db)$. \qed
\end{lemma}

%%%%%%%%%%%%%%%%%%%%%%%%%%%%%%%%%%%%%%%%%%%%%%%%%%%%%%%%%%%%%%%%%%%%%%%%%%%%%%%
\subsection{First step towards the query\texorpdfstring{ $\psineg$.}{.} }

\noindent
Now we are ready for:

\begin{specification}\label{spec:sensors}\label{spec:4}
The set $\{\pphi_i(x_i)\}_{i\in\III}$ will satisfy:

\vspace{1mm}\noindent
(con) For each $i \in \III$ the query $\pphi_i$ is connected.

\vspace{1mm}\noindent
(pos) if $\thue$ is a positive instance of the problem $\Thue$, then for all  $\db$ such that  $\db\models \newaboxneg$ and $\db\models \ttbox $ there exists $i\in \III$ such that $\dblove \models \jest{\pphi_i}$;

\vspace{1mm}\noindent
(neg) if $\thue$ is a negative instance of the problem $\Thue$, then
 there exists a structure $\dbmbb$ such that $\dbmbb\models \newaboxneg$ and $\dbmbb\models \ttbox$  but
for each $i\in \III$
it holds that $\dbmbb\not\models \jest{\pphi_i}$. Additionally, $\dbmbb\not\models \exists x\, Pillow(x)$.
\end{specification}

The intuition is that the queries $\jest{\pphi_i}$, for ${i\in\III}$,  collectively probe the structure $\dblove$ 
in order to detect positivity of $\thue$ there.

The queries (and also the set $\III$) will be defined in Section \ref{sec:low-bears-beds}. But let us already reveal here that all of them but one will
be conjunctive queries with safe negation, and one will be a CQ.
Then, the idea is that the more $\dblove$ will look like an infinite tree,
the more the queries with safe negation will be likely to be satisfied in $\dblove$. And
the more $\dblove$ will
look like the aforementioned a self-looping $a$, the more the single CQ will be likely to be satisfied.

Notice that if (as in (pos) above) there is $\dblove \models \jest{\pphi_i}$ then also $\db \models \jest{\pphi_i}$,
via the same mapping. This is true, despite the fact that
queries $\jest{\pphi_i}$ may contain negated atoms,
because $\dblove$ is defined
as the {\bf induced} (by the set of vertices of $\love(\db)$) substructure of $\db$.

There will be also one more very simple CQ needed in our construction:
$\pphi(x) =  Pillow(x)$.

\subsection{Little bears and their beds.}\label{sec:little-bears-and-beds}

\begin{specification}[beds, part 2]\label{spec:5}
For all $i,j\in \III$:
\begin{enumerate}[label=(\alph*),itemsep=0em]
    \item  $\newaboxneg_i\models \ttbox $;
    \item  $a_i$ is a vertex of $\newaboxneg_i$ and $\newaboxneg_i\models 2Close_i(a_i,a_i)$;
    \item  there exists a mapping $g_i$ satisfying $\jest{\pphi_{i}}$ in  $\newaboxneg_i$;
     if $\db \models \ontolneg$ then the above mapping  $g_i$   also satisfies $\jest{\pphi_{i}}$ in $\db$; 
     \item $\newaboxneg_i \models\! Pillow(a_i)$
    \item  if $f$ is any mapping  satisfying $\jest{\pphi_{i}}$ in  $\newaboxneg_j$ then $f(x_i)=a_j$ (recall that $x_i$ is  the distinguished variable  of  $\jest{\pphi_{i}}$);
    \item  if $\newaboxneg_i\models Pillow(c)$ for some vertex $c$ then  $c=a_i$.\qedhere
\end{enumerate}
\end{specification}

Imagine that queries $\pphi_i$, for $i\in \III$ are little bears, and  the structures $\newaboxneg_i$ 
as their designated beds. And imagine the query $\pphi$ as Goldilocks, who does not really care where she sleeps,
except that she needs a pillow.

The (c) above says that, whatever happens (that is, whatever model of $\ttbox$ we consider), each of the bears 
can sleep in his dedicated  bed. The (d) says that (whatever happens) also Goldilocks can sleep in any of the beds. 

The (e) and (f) say that there is only one right way, for each of our characters, to sleep in a bed: their head (the distinguished variable) must be mapped onto the pillow (the $a_i$ vertex of the $\newaboxneg_i$).

A curious Reader could wonder here, how can we possibly define the structures in such a way that (c) is satisfied? However we defined $\newaboxneg_i$, the $\db$ will be some unpredictable 
superstructure of $\newaboxneg_i$, and how can we be sure that the negated atoms of $\pphi_i$
are still satisfied in $\db$ (which means, that the relevant positive atoms remain false)? The answer is, that this requires all the objects here including the queries $\pphi_i$ and
the structures $\newaboxneg_i$ to play as a one  well-coordinated orchestra, with the $\freezed$ operator  playing the first violin in it.

The next lemma will be soon very useful, and it also gives us an opportunity to see  the $\freeze$ operator in action:

\begin{lemma}\label{jeden-hejt}
If $\db\models \ontolneg$, if $i\neq j$ for some $i,j\in \III$ and if $k\in\III$, then
$\db\not\models 2Close_k(a_i,a_{j})$.
\end{lemma}

For the proof of the Lemma suppose $k\neq i$. Then $a_i$ is not involved in any atom of $2Close_j$ in $\newaboxneg_i$.
So the statement $\freeze{\bigcup_{i\in \III} \newaboxneg_i}$ in $\ontolneg$   guarantees that $a_i$ is not involved in any
atom of $2Close_k$ also in $\db$. The case
when $k\neq j$ is analogous. \qed

\subsection{The query \texorpdfstring{$\psineg$}{Q} and the main idea.}\label{sec:psinek}

\begin{definition}

$\psineg$ is 
%defined as
the existential closure of the query:
$$ \pphi(x) \;\; \wedge \; \bigwedge_{i\in\III} ({\pphi_i}\;\wedge \neg 2Close_j(x,x_i)) \wedge\; $$
$$ \wedge \bigwedge_{i,i'\in \III, i\neq i'}\bigwedge_{j\in \III} (\neg 2Close_j(x_i,x_i')\wedge
\neg 2Close_j(x_i,x)) \quad \qedhere$$
\end{definition}

The idea here is that there are $|\III|$+1 characters that need to be put to bed. We have $|\III|$ beds available,
but there may also be an extra place for one of the bears to sleep, koala\footnote{The authors are perfectly aware of the fact that koalas are not bears.} way, somewhere inside the 
``tree'' $\db[\love(\db)]$ -- the existence of this
place is an undecidable property of $\thue$ (recall Specification \ref{spec:sensors}). And there is some mechanism in place
(namely the part involving the $2Close$ predicates) that is supposed to prohibit putting more than one character to the same bed.

Now, once we have both $\ontolneg$ and $\psineg$ specified, we are able to prove the equivalence $\bigpumpkin$ 
from Section \ref{subsec:source}.

\subsection{The \texorpdfstring{$\Leftarrow$}{<=} direction.}

Suppose $\thue$ is a negative instance of $\Thue$. We just need to construct one structure $\dbmf$, such that $\dbmf\models\ontolneg$
but $\dbmf\not\models\psineg$.

This $\dbmf$ will be defined as a disjoint union of structures:
$$ \dbmf \;\; = \;\;  \dbmbb \; \cupdot  \; \bigcupdot_{i\in \III}\; \newaboxneg_i $$
Where $\dbmbb $ is as in Specification \ref{spec:sensors} (neg).

First of all we need to show that $\dbmf\models\ontolneg$. But clearly, all the facts from $\newaboxneg \cup \bigcup_{i\in \III} \newaboxneg_i$  are true in $\dbmf$ (recall that $\dbmbb \models \newaboxneg$). Also the conditions imposed by $\freeze{\bigcup_{i\in \III} \newaboxneg_i}$ are obviously
satisfied
in
$\dbmf$.
This is because there are no new atoms in $\dbmf$, involving any of the constants from
$\dom{\bigcup_{i\in \III} \newaboxneg_i}$,
compared to $\bigcup_{i\in \III} \newaboxneg_i$ itself. What concerns $\ttbox$, the structure  $\dbmf$ is a disjoint union of several structures and it is enough to verify for each of these structures that it satisfies $\ttbox$.  But we know that
$\dbmbb\models  \ttbox$
(Specification \ref{spec:sensors} (neg)) and we also know that each $i\in\III$ it holds that
$\dbmbb \models \newaboxneg$
(Specification \ref{spec:5} (a)).

Now we need to prove that $\dbmf\not\models\psineg$. So suppose $\dbmf\models\psineg$. What would it mean exactly?
It would mean that a set of mappings $\{h\}\cup \{h_i\}_{i\in \III}$ could be found, such that:\\
\vspace{1mm}
\textbullet~$h_i$ satisfies $\jest{\pphi_i}$ in $\dbmf$ for each $i$,\\
 \textbullet~ $h$ satisfies $\jest{\pphi}$ in $\dbmf$,\\
\vspace{1mm}
\noindent
and such that neither 
$h(x)$ is  too close to any of the  $h_i(x_i)$ nor   $h_i(x_i)$  is too close to any $h_{i'}(x_{i'})$ 
for   $i\neq i'$.

%%%%%%%%%%%%%%%%%%%%%%%%%%%%%%%%%%%%%%%%%%%%

Recall that each of $\pphi_i$, for $i\in \III$, is a connected query
(by \ref{spec:4} (con)). So the entire image of $h_i$ must fall into one of the connected components of $\dbmf$.
But it cannot fall into $\dbmbb$ (by Specification \ref{spec:sensors} (neg)). Also 
the image of $h$ cannot fall into  $\dbmbb$ (see  Specification \ref{spec:sensors} (neg) again).
Apart from $\dbmbb$ the structure $\dbmf$ has $|\III|$ connected components: the beds 
$\newaboxneg_{i}$ for $i\in \III$. So, by the bearhole principle,
images of at least two of the aforementioned $|\III|+1 $
mappings must fall into some $\newaboxneg_{k}$. In other words, two characters must be put to the same bed.
But, by Specification \ref{spec:5} (e) and (f) this would mean that their heads (some $x_i$ and $x_j$, or $x_i$ and $x$) are mapped to the 
same pillow $a_k$. Which would mean that the constraint $\neg2Close_k(x_i,x_j)$
(or $\neg2Close_k(x_i,x)$)
would 
not be satisfied.

%%%%%%%%%%%%%%%%%%%%%%%%%%%%%%%%%%%%%%%%%%%%%
\subsection{The \texorpdfstring{$\Rightarrow$}{=>} direction.}

Now suppose that $\thue$ is a positive instance of $\Thue$ and that $\db\models\ontolneg$. We  need to
show that $\db\models\psineg$.

Again, this means that we need to find a set of mappings $\{h\}\cup \{h_i\}_{i\in \III}$, such that
$h_i$ satisfies $\jest{\pphi_i}$ in $\db$
for each $i$
(and $h$ satisfies $\jest{\pphi}$ in $\db$) and such that neither
$h(x)$ is too close to any of the  $h_i(x_i)$, nor   $h_i(x_i)$ is too close to $h_{i'}(x_{i'})$ for any $i\neq i$.

So let us first use the assumption that  $\thue$ is a positive instance. It implies (by Specification \ref{spec:sensors} (pos))
that there exists $\iii\in \III$ and a mapping $h_\iii$ which satisfies $\jest{\pphi_\iii}$ in $\db$, and such that
$h_\iii(x_\iii)\in \love(\db)$ (in fact, $h_\iii$ maps all variables of $\jest{\pphi_\iii}$ to $\love(\db)$, but it is
only the image of $x_\iii$ that really matters to us).

Recall that (by Lemma \ref{lem:no-hatred}) vertex $h_\iii(x_\iii)$ is not too close to anything. So in
particular $h_\iii(x_\iii)$ may not happen to be equal to any $a_j$, for $j\in \III$, because each such $a_j$ is 
too close to itself (by Specification \ref{spec:5} (b)).

Now we are ready to satisfy $\jest{\pphi}$. We put $h(x)=a_\iii$ (recall Specification \ref{spec:5} (d)). See, we do not need this $a_\iii$ any more for
$x_\iii$, because $x_\iii$ is  already happily mapped to some vertex from $\love(\db)$.

As the next step, we will satisfy all the queries $\jest{\pphi_i}$ for $i\neq \iii$. This is easy: just put $h_i=g_i$
 (recall Specification \ref{spec:5} (c)).

 Now we need to notice that neither $h_i(x_i)$ is too close to $h_j(x_j)$, for some $i\neq j$,
 nor any $h_i(x_i)$ is too close to $h(x)$. But, as we said above, $h_\iii(x_\iii)$ is never too close. And in the remaining cases
 the two variables in question are mapped to two some $a_i$ and $a_j$, for $i\neq j$ and we use Lemma \ref{jeden-hejt}. \qed

 This ends the proof of the equivalence $\bigpumpkin$ and hence also the proof of Theorem \ref{mainp}. Modulo, of course,
 an implementation of the specifications.

\section{Definitions}\label{sec:implementation}

In this section we define the objects listed as (i)--(iv) at the beginning of Section \ref{sec:high-level}

\subsection{Source of Undecidability}

Our source of undecidability will be, a specially tailored for this occasion, variant of the word problem for finitely generated semigroups\footnote{This problem is also known as the word problem for Thue systems}. 

\begin{definition}[source of undecidability]\label{def:source}
Decision problem  $\Thue$ is defined as follows.

An  instance $\thue$ of $\Thue$ is a pair $\pair{\MMM, \Pi}$, where:

\begin{itemize}[topsep=-0.0em]

\item
$\MMM = \MMM_0 \cup \set{\es, \esp} $ is a finite set of \emph{symbols}, such that $\es,\esp \not\in \MMM_0$, 
and  that $ \start, \finish, \blank \in \MMM_0$;

\item $\Pi \subseteq \MMM^* \times \MMM^*$ is a {\bf symmetric} relation, and  

\item $\Pi$ is a union of three disjoint sets $\Pi_{22}$, $\Pi_{21}$, and $\Pi_{12}$ s.t.:

\begin{enumerate}[label=(\roman*), topsep=-0.1em]
    \item $\Pi_{22} \subseteq \MMM_0^2 \times \MMM_0^2$, and 
    \item $\Pi_{21} =\set{\pair{\blank\esp,  \es} , \pair{\blank\es,  \esp}}$;
    \item $\Pi_{12} =\set{\pair{\es, \blank \esp}, \pair{\esp, \blank \es}}$;
   
   \item if  $\pair{\cpred\dpred,\epred\fpred} \in \Pi_{22}$, for some symbols $\cpred,\dpred,\epred,\fpred\in \MMM_0$,  then $\dpred\neq \fpred$.
\end{enumerate}
\end{itemize}
We define $\szym$ as the relation satisfying $wlv \szym wrv$, for all
$w,v \in \MMM^*$ and every pair $\pair{l,r} \in \Pi$.
Since $\Pi$ is symmetric, the relation $\szym$ is symmetric as well. Finally, we define
$\szymeq$ as the smallest equivalence relation containing $\szym$.
An instance  $\thue  $ of $\Thue$ is {\em positive} if $\start\es\szymeq \finish\es$ and is {\em negative} otherwise.
\end{definition}

%This resolves \cref{prom:of-problem}. 
The following observation can be obtained via a standard reduction from the Turing Machines  halting problem. Notice that 
$\blank$ is the blank symbol, and  $\es$ and $\esp$ are two end-of-the-tape symbols, able to produce any number of blanks. Because of certain technical nuance we need two of them (see the proof of Lemma \ref{lem:krolik}). The pairs from 
$\Pi_{22}$ simulate the moves of the machine head. Notice that this means 
that, if $\pair{\cpred\dpred,\epred\fpred} \in \Pi_{22}$ then exactly one 
of the symbols $\dpred$, $\fpred$ encodes the machine head, and in consequence we indeed 
 have that $\dpred\neq \fpred$.

\begin{observation}
    The problem, whether a given instance of $\Thue$ is positive, is undecidable and recursively enumerable.
\end{observation}

From now on,  we fix an instance $\thue$ of $\problem$. 
Note that this means that $\MMM$ is also fixed.

Before we proceed towards definitions of the remaining objects 
let us recall the following well-known fact:

\begin{fact}\label{fact}
 For a set of rewriting rules $\Pi$, and for ${w},{w}'\in  \MMM^*$ we have ${w}\szymeq {w}'$ {\iffi} there is a finite path from ${w}$ to ${w}'$ in the infinite undirected graph, whose 
 vertices are words from $\MMM^*$, and whose set of edges is the relation $\szym$.
 \end{fact}

%%%%%%%%%%%%%%%%%%%%%%%%%%%%%%%%%%%%%%%%%%%%%%%

\subsection{Our little bears: definition of queries \texorpdfstring{$\beari$}{A\_i}. }\label{sec:low-bears-beds}

\begin{definition}
 The set   $\III$ is defined as $\Pi \cup \set{\diam}$. 
\end{definition}

For each name\footnote{We like to think of elements of $\III$ as \emph{ names} of our bears.}  $i\in \III$ we shall now define the query $\beari$, with one free variable $x_i$. Notice that there are four kinds of elements in $\III$ and for this reason we will consider four cases (see \cref{fig:all-queries} for visualization):

\begin{figure}
    \centering
    \includegraphics[width=1.0\linewidth]{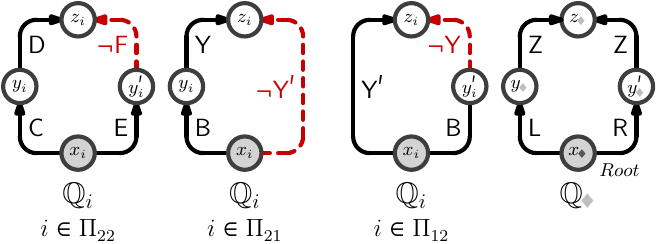}
    \caption{From left to right: $\beari$ for when: (1) $i = \pair{\cpred\dpred,\epred\fpred} \in \Pi_{22}$, (2) $i = \pair{\blank\ypred,\ypred'} \in \Pi_{21}$, (3) $i = \pair{\ypred',\ypred\blank} \in \Pi_{12}$, and (4) $i = \idiam$. Red dashed arrows with negated labels indicate negated atoms. All variables are existentially quantified, except those with a gray background.}
    \label{fig:all-queries}
\end{figure}

\noindent
{\textbf{Case}} $i=\diam$. Then
$\bear_\idiam(x_\idiam) = $\\
$\exists y_\idiam\!,y_\idiam'\!, z_\idiam\; Root(x_\idiam),\start(x_\idiam, y_\idiam), \zpred(y_\idiam,z_\idiam), \finish(x_\idiam, y_\idiam'), \zpred(y_\idiam',z_\idiam).$

\medskip
\noindent
\textbf{Case} $i\in \Pi_{22}$. Then  $i$ {is of the form} $\pair{\cpred\dpred,\epred\fpred}$ for some 
$\cpred,\dpred,\epred,\fpred\in \MMM_0$, and:\\
$\bear_i(x_i) = \exists y_i,y_i', z\;\; \cpred(x_i, y_i), \dpred(y_i,z_i),
                                     \epred(x_i, y_i'), \neg \fpred(y_i',z_i).$

\medskip
\noindent
\textbf{Case} $i\in \Pi_{21}$. Then $i$ is of the 
form $\pair{\blank\ypred,\ypredp}$ for some $\{\ypred, \ypredp \}=\{\es,\esp\}$, and:\\
\phantom{a} \hfill
$\bear_i(x_i) = \exists y_i, z_i\;\;\; \blank(x_i,y_i), \ypred(y_i,z_i),\neg \ypredp(x_i,z_i)$\hfill\phantom{a}

\medskip
\noindent
\textbf{Case} $i\in \Pi_{12}$. Then $i$ is of the form $\pair{\ypredp, \blank\ypred}$ for some $\{\ypred, \ypredp \}=\{\es,\esp\}$ and:\\
\phantom{a}\hfill
$\bear_i(x_i) = \exists y_i, z_i\;\;\; \blank(x_i,y_i), \neg \ypred(y_i,z_i), \ypredp(x,z)$
\hfill\phantom{a}

We provide visualization for the queries below:

\subsection{Specification \texorpdfstring{\ref{spec:4}}{of queries}.}

Having the queries $\beari$ defined, for each $i\in \III$, we need to prove that
 \cref{spec:4} is indeed satisfied. It is clear that all the defined queries are connected. 
 What remains to be proved are  negative (neg) and positive (pos) parts of \cref{spec:4}. This part 
 is analogous to the respective part of the construction in \cite{MO25}.
 %\orange{omówić} 
 For the sake of completeness however
 we present the argument in Appendix A.
 
%%%%%%%%%%%%%%%%%%%%%%%%%%%%%%%%%%%%%%%%%%%%%%%%%%%%%%%%%%%%%%%
\subsection{Bears want to sleep. The structures \texorpdfstring{$\newaboxneg_i$}{A\_i}.}

Definitions of the structures $\newaboxneg_i$, for $i\in \MMM$, are best presented as a figure. From left to right $\newaboxneg_i$ for: (1) $i = \pair{\cpred\dpred,\epred\fpred} \in \Pi_{22}$, 
(2) $i = \pair{\blank\ypred',\ypred} \in \Pi_{21}$, (3) $i = \pair{\ypred',\ypred\blank}   \in \Pi_{12}$ (for some 
$\{\ypred,\ypred'\}=\{\es,\esp\})$,
and (4)  $i = \diam$. Unlabelled (rainbow) edges indicate the whole set $\MMM$; predicate $\pillow$ is indicated by a rhomboid outline. Notice that \cref{spec:of-disjointness} is  satisfied.

\begin{center}
    \includegraphics[width=1.0\linewidth]{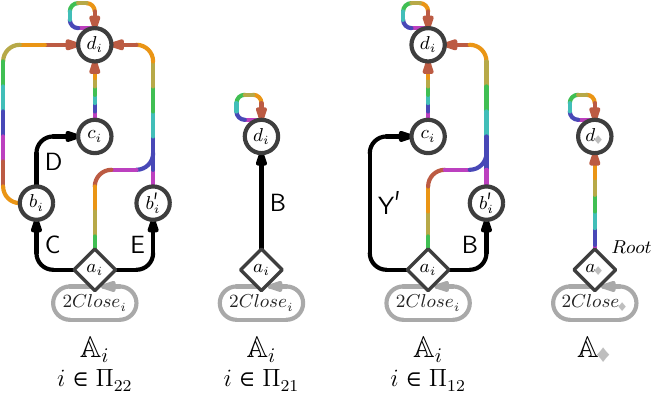}
\end{center}

%%%%%%%%%%%%%%%%%%%%%%%%%%%%%%%%%%%%%%%%%%%%%%%%%%%%%%%%%%%%%%%%%%%%%%%%%%
% VERIFICATION OF SPECIFICATIONS
%%%%%%%%%%%%%%%%%%%%%%%%%%%%%%%%%%%%%%%%%%%%%%%%%%%%%%%%%%%%%%%%%%%%%%%%%%
\section{Specification \texorpdfstring{\ref{spec:5}}{of Beds pt. 2}}\label{sec:dowody}
In this section we verify that Specification~\ref{spec:5} is indeed satisfied by our queries and structures. 

Regarding the conditions (b), (d) and and (f) of \cref{spec:5} there is nothing to prove. 
Let us now deal with (a):

\begin{observation} $\newaboxneg_i\models \ttbox $ for all names $i\in \III$.
\end{observation}
\begin{proof} 
        
            \begin{enumerate}[itemsep=0em]
                \item[(h1)] Every vertex of each $\newaboxneg_i$ except for $a_i$ has an 
                outgoing rainbow edge, denoting all $\spred \in\MMM$. And none of the $a_i$ has any incoming edges from $\MMM$.
                \item[(h2)] The only vertex that satisfies $Root$ is $a_\idiam$, and $\newaboxneg_\idiam \models \spred(a_\idiam, d_\idiam)$ for each $\spred \in \MMM$.    
                \item[(h3)] Only $\pair{a_i,a_i}$ satisfies $\tooclose_j$ (for any $j$) in any $\newaboxneg_i$. But none of the $a_i$ has any incoming edges from $\MMM$.  And also  $\unaryA$ is not satisfied anywhere in any  $\newaboxneg_i$.    \qedhere
            \end{enumerate}
    
\end{proof}

The (c) of  \cref{spec:5} comes next:

\newcommand{\apsto}{\!\mapsto\!}\label{lem:krolik}
\begin{lemma} \textbullet~ For each $i \in \III$ there exists a mapping $g_i$ satisfying $\jest{\pphi_{i}}$ in  $\newaboxneg_i$. \\
\textbullet~ If $\db \models \ontolneg$ then the above mapping  $g_i$ also   satisfies $\jest{\pphi_{i}}$ in $\db$.

\end{lemma}
\begin{proof} We begin by defining $g_i$:
    \begin{itemize}[topsep=-0.1em, itemsep=-0.1em]
        \item \parbox{2.5cm}{If $i \in \Pi_{22}$ then $g_i$} $ = \set{x_i \apsto a_i, y_i \apsto b_i, y_i' \apsto b_i', z_i \apsto c_i}$
        \item \parbox{2.5cm}{If $i \in \Pi_{21}$ then $g_i$} $= \set{x_i \apsto a_i, y_i \apsto d_i, z_i \apsto d_i}$
        \item \parbox{2.5cm}{If $i \in \Pi_{12}$ then $g_i$} $= \set{x_i \apsto a_i, y_i' \apsto b_i', z_i \apsto c_i}$
        \item \parbox{2.5cm}{If $i = \diam$ then $g_\idiam$} $ \!\!\!\!=\set{x_\idiam \apsto a_\idiam, y_\idiam \apsto d_\idiam, y_i' \apsto d_\idiam, z_i \apsto d_\idiam}.$
    \end{itemize}

    It is straightforward to verify that, in each case, the above defined $g_i$  is indeed a mapping 
    satisfying $\jest{\pphi_{i}}$ in  $\newaboxneg_i$.

    For the second claim, take any $\db$ such that $\db \models \ontolneg$. It is here where {\bf we are  finally going to use the $\freezed$
    operator!}

    Clearly, $g_i$ is a homomorphism from the positive part of $\bear_i $ to 
    $\newaboxneg_i$, and (since  $\newaboxneg_i \subseteq \db$ ) it is also a homomorphism
    from the positive part of $\bear_i $ to $\db$.

    The tricky part is to make sure that also the constraint imposed by the negative atom from  $\newaboxneg_i$ 
   remains satisfied in  $\db$. As usual, we need to consider four cases:

    \begin{itemize}[topsep=-0.1em, itemsep=-0.1em, itemindent=-2mm ]
    \item $i = \diam$. Then there is no negative atom in $\newaboxneg_i$, so there is nothing to prove. 
        
    \item $i = \pair{\cpred\dpred,\epred\fpred} \in \Pi_{22}$. We need to show that 
    $\db\not\models \fpred(b_i', c_i)$. But $\newaboxneg_i \not\models \exists{x}\,\fpred(x,c_i)$ (recall that 
    we assumed, in Definition \ref{def:source}, that $\fpred\neq \dpred$)
    so (since the types of vertices of  $\newaboxneg_i$ are frozen by the condition $\freeze{\aabox_0}$
    in the definition of $\ontolneg$, we get that $\db\not\models \exists{x}\,\fpred(x,c_i)$.

        \item  $i = \pair{\blank\ypred, \ypred'} \in \Pi_{21}$.  
      Then we need to show that $\db\not\models \ypred'(a_i,d_i)$.
         But $\newaboxneg_i \!\not\models\! \exists{x}\,\ypred'(a_i,x)$, and we use the freezing argument again. 
        \item  $i = \pair{\ypred',\blank\ypred} \in \Pi_{12}$. 
        We need to show that 
        $\db\not\models \ypred(b_i', c_i)$. But $\newaboxneg_i \not\models \exists{x}\,\ypred(x,c_i)$ as 
        (recall that, by Definition \ref{def:source}, $\ypred \neq \ypred'$) and again the freezing argument applies. 
    \end{itemize}
   \end{proof}

\subsection{Tall Bears and Their Short Beds. }

The last condition to be verified, from Specification \ref{spec:5}, is (e):

 \begin{lemma}\label[lemma]{lem:last-spec-lemma} For each $i,j \in \III$ if $f$ is any mapping  satisfying $\jest{\pphi_{i}}$ in  $\newaboxneg_j$ then $f(x_i)=a_j$.
 \end{lemma}

\noindent
{\em Proof:} If $i=\diam$ then $\jest{\pphi_{i}}$ requires that $Root(f(x_\diam))$, and there is nothing to prove then, because $a_\diam$ is the only vertex in any $\newaboxneg_j$ satisfying this requirement.

So fix $i\in \Pi$ and $j\in\III$,  and $f$, as in \cref{lem:last-spec-lemma} and
assume, towards contradiction, that  $f(x_i)\neq a_j$. Then, of course:

\begin{observation}
    $f(u)\neq a_j$ ~holds for each variable $u$ of $\jest{\pphi_{i}}$.
\end{observation}

Let now $\neg\spred(u,z_i)$ be the negative atom in $\jest{\pphi_{i}}$.
Notice that $\spred(f(u),d_j)\in\newaboxneg_j$.
This is because  $f(u)\neq a_j$, and because whatever $j$ we took, there is a rainbow edge from each vertex of $\newaboxneg_j$ (except possibly for $a_j$) to $d_j$. We have just proved:

\begin{observation}
    $f(z_i)\neq d_j$.
\end{observation}

Now we will use the fact that $f$ maps paths in $\beari$ to paths in $\newaboxneg_j$.

\begin{definition}
    For $\bear_k$  let $\size(\bear_k)$ be the {minimal} number of vertices on a positive (using only non-negated atoms) path from $x_k$ to $z_k$.  
    For  $\newaboxneg_k$ let $\size(\newaboxneg_k)$ be the {maximal} number of vertices on a  path from some $s\neq a_k$ to some  $t\neq d_k$. 
    %We use a function $\size(\cdot)$ to extract the size from an object.
\end{definition}

It is very easy to verify that:

\begin{observation}\label{obs:bear-bed-sizes}
$\size(\newaboxneg_\diam)= 0$. 
%and $\size(\bear_\diam)= 2$. 
And:\\
if $k \in \Pi_{22} $ then $\size(\newaboxneg_k)= 2$ and $\size(\bear_k)= 3$;\\
if $k \in \Pi_{21} $ then $\size(\newaboxneg_k)= 0$ and $\size(\bear_k)= 3$;\\
if $k \in \Pi_{12} $ then $\size(\newaboxneg_k)= 1$ and $\size(\bear_k)= 2$.
 \end{observation}

Since, as we said, $f$ maps paths in $\beari$ to paths in $\newaboxneg_j$ we get:

\begin{observation}\label{ob:29}
    $\size(\bear_i) \geq \size(\newaboxneg_j)$.
\end{observation}

It follows immediately from Observations \ref{obs:bear-bed-sizes} and \ref{ob:29} that
$i\in \Pi_{12}$ and $j\in \Pi_{12}$.
So let $i$ be  $\pair{\ypred', \blank\ypred}$ and let $j$ be  $\pair{\cpred\dpred, \epred\fpred}$. Recall that $f(z_i) \neq d_j$. Therefore $f(x_i) = b_j$ and $f(z_i) = c_j$. But $\ypredp \neq \dpred$ by definition of $\Pi$.
Contradiction. \qed

\newpage
%% The file kr.bst is a bibliography style file for BibTeX 0.99c
\bibliographystyle{abbrv}
\bibliography{bibliography}

%%%%%%%%%%%%% APPENDIX A

 \section{Appendix A}
 
\subsection{Negative Part}
\begin{lemma}\label[lemma]{lem:spec-5-neg}
If $\thue$ is a negative instance of the problem $\Thue$, then
there exists a structure $\dbmbb$ such that $\dbmbb\models \newaboxneg$ and $\dbmbb\models \ttbox$  but for each $i\in \III$
it holds that $\dbmbb\not\models \jest{\pphi_i}$. Additionally, $\dbmbb\not\models \exists x\, \pillow(x)$.
\end{lemma}

We devote the rest of this subsection to the proof of \cref{lem:spec-5-neg}. To this end we first define structure $\dbmbb$ over $\Sigma$.

\begin{definition}
The set $\VVV$ of vertices of $\dbmbb$ is the set of all equivalence classes of the relation $\szymeq$, that is $\VVV = \MMM^*/{\szymeq}$.

\smallskip
\noindent
For each $w,v \in \MMM^* $ and each symbol $\spred \in \MMM$ define:
$$\dbmbb\models \spred([w]_{\szymeq}, [v]_{\szymeq})  \quad\iffi\quad w\spred \szymeq v$$
One needs to notice here that the above definition
is independent of the choice of representatives $w$ and $v$, from their respective  equivalence classes.
So assume that  $w\spred \szymeq v$ and take $w',v'$ such that $w\szymeq  w'$ and $v\szymeq  v'$. Then:
$$w'\spred\; \szymeq \;  w\spred \; \szymeq \; v \; \szymeq \; v'$$
\noindent
Define also:
$$\dbmbb\models \korz([w]_{\szymeq})  \quad\iffi\quad  w \szymeq \varepsilon$$

and 

$$\dbmbb\models \unaryA([w]_{\szymeq})  \quad\iffi\quad  w \szymeq \varepsilon$$
\noindent
where $\varepsilon$ denotes the empty word.
Finally, let:
$ a=[\varepsilon]_{\szymeq}.$
\end{definition}

By definition, we clearly get:
\begin{observation}
    $\dbmbb \models \newaboxneg$ and  $\dbmbb \not \models \exists{x}\, Pillow(x).$ 
\end{observation}

The remaining two conditions require simple proofs:
\begin{lemma}
    $\dbmbb \models \ttbox.$
\end{lemma}
\begin{proof}
    By definition, each vertex of $\dbmbb$ has an outgoing $\spred$ edge for each $\spred\in\MMM$, that is for each $\szymeqrep{w} \in \VVV$ and $\spred \in \MMM$, we have $\dbmbb \models \spred(\szymeqrep{w},\szymeqrep{w\spred})$. Thus both (i) and (ii) of $\ttbox$ (Definition \ref{def:tbox}) are satisfied.

    We get (iii) by noting that for all $i \in \III$ the relation $\tooclose_i$ does not appear in the definition of $\dbmbb$.
\end{proof}

\begin{lemma}
For each name $i$ it holds that $\dbmbb\not\models \jest{\beari}$.
\end{lemma}
\begin{proof}
    Assume towards contradiction that $\dbmbb\models \jest{\beari}$. Let $f$ be the satisfying mapping from $\jest{\beari}$ to $\dbmbb$. We consider two cases.
    
    First, take $i \in \Pi$. Consider the case when $i$ is of the form $\pair{\cpred\dpred,\epred\fpred}$.  Let $w \in \MMM^*$ be a word such that $f(x_i) = [w]_{\szymeq}$, then $f(y'_i) = [w\epred]_{\szymeq}$ and $f(z_i) = [w\cpred\dpred]_{\szymeq}$ by definition. However $w\epred\fpred \szymeq w\cpred\dpred$ as $i \in \Pi$. Therefore $\dbmbb \models {\fpred}([w\epred]_{\szymeq},[w\cpred\dpred]_{\szymeq})$ and so $\dbmbb \models {\fpred}(f(y'_i), f(z_i))$, but $\jest{\beari} \models \neg\fpred(y'_i, z_i)$, which is  a contradiction. The other cases when $i \in \Pi$ are analogous.
    
    Second, take $i = \diam$. Note that $f(x_\diam) = a$, since by definition $Root$ predicate is satisfied only in the vertex $a$ of $\dbmbb$. Consider $f(z)$, clearly we have $[\LL\esp]_{\szymeq} = f(z) = [\RR\esp]_{\szymeq}$. Therefore $\LL\esp \szymeq \RR\esp$, which contradicts the assumption that $\thue$ is a negative instance.
\end{proof}

\subsection{Positive Part}
\begin{lemma}
If $\thue$ is a positive instance of the problem $\Thue$, then for every  $\db$ such that  $\db\models \newaboxneg$ and $\db\models \ttbox $ there exists a name $i$ such that $\dblove \models \jest{\beari}$;
\end{lemma}

Assume $\thue$ is positive, and fix a structure $\db$ such that $\db\models \newaboxneg$ and $\db\models \ttbox$. Consider $\dblove$, and recall that the set of vertices of $\dblove$ is denoted with $\vlove$.

We introduce a useful notation for reachability in $\dblove$.
\begin{definition}
For two vertices $s,t$ of $\vlove$ and for a symbol $\spred \in \MMM$  we write $\strz{s}{\spred}{t}$ if $\dblove \models \spred(s,t)$, and for a word $w\in \MMM^*$ we write $\strz{s}{w\spred}{t}$ if there exists $u\in \vlove$ such that $\strz{s}{w}{u}$ and $\strz{u}{\spred}{t}$.
\end{definition}

\begin{definition}[Perfect Structure]
    We say $\dblove$ is \emph{perfect} if\\ $\;\;\strz{s}{l}{t}\;\; {\iffi}  \;\;\strz{s}{r}{t}$, for all $s,t \in \vlove$, and $\pair{l,r} \in \Pi$.
\end{definition}

The following lemma is natural:

\begin{lemma}\label{lem:perfection} 
Take  $s\in \vlove$ and $w,v\in \MMM^*$, if $\dblove$ is perfect and
 $\strz{a}{w}{s} \land w \szymeq v $  then  $\strz{a}{v}{s}$.
\end{lemma}
\begin{proof}
By Fact \ref{fact} and induction with respect to the number of rewriting steps needed to produce $v$ from $w$.
\end{proof}

\begin{corollary}
    If $\dblove$ is perfect then $\dblove \models \bear_\diam(a).$
\end{corollary}

\begin{lemma}
    \!If $\dblove$\! is imperfect then $\dblove \!\models\! \jest{\beari}$ for some \small $i \!\in\! \III \setminus \set{\!\diam\!}$.
\end{lemma}
\begin{proof}
    Take some $s,t \in \vlove$ and $\rho = \pair{l,r} \in \Pi$ violating perfection of $\dblove$. Assume w.l.o.g that $\strz{s}{l}{t}$ but not $\strz{s}{r}{t}$. Note that $r = \spred$ or $r = \ppred\spred$ for some symbols $\spred,\ppred \in \MMM$ by definition --- that is $r$ is of length 1 or 2. If $r = \spred$ then $\dblove \not\models \\spred(s,t)$ by assumption, thus trivially $\dblove \models \bear_\rho(s)$. If $r = \ppred \spred$ then note that there exists a vertex $u$ such that $\dblove \models \\ppred(s,u)$ as $\dblove$ is a model of $\ttbox$. However $\dblove \not\models \\spred(u,t)$, by assumption, thus $\dblove \models \bear_\rho(s)$.
\end{proof}

%\newpage
%\input{X1-Appendix}

\end{document}